\pdfoutput=1
\documentclass[letterpaper, 10 pt, conference]{ieeeconf}  % Comment this line out if you need a4paper

\IEEEoverridecommandlockouts                              % This command is only needed if 
\usepackage{amsmath,amsfonts}
\usepackage[algo2e,ruled,vlined,linesnumbered,resetcount,norelsize]{algorithm2e}
\usepackage{array}
\usepackage[caption=false,font=normalsize,labelfont=sf,textfont=sf]{subfig}
\usepackage{textcomp}
\usepackage{stfloats}
\usepackage{float}
\usepackage{url}
\usepackage{verbatim}
\usepackage{graphicx}
\usepackage{cite}
\usepackage{mathtools} % for coloneqq
\usepackage{amssymb}  % assumes amsmath package installed

\usepackage{amsthm}
\usepackage{algpseudocode}
\usepackage{xcolor}
\usepackage{accents}
\usepackage{mathrsfs}
\usepackage[normalem]{ulem}

\def\beq{\begin{equation*}}
\def\eeq{\end{equation*}}
\def\bql{\begin{equation}}
\def\eql{\end{equation}}
\def\bqn{\begin{eqnarray*}}
\def\eqn{\end{eqnarray*}}
\def\bnl{\begin{eqnarray}}
\def\enl{\end{eqnarray}}

\usepackage{soul}
\usepackage{booktabs}

\usepackage[T1]{fontenc}
\usepackage[utf8]{inputenc}
\usepackage{siunitx} % for 'S' column type
\newcolumntype{T}{S[table-format=3.3, input-symbols={()},
                    table-space-text-post={$^{***}$},
                    table-align-text-post=false]}
\usepackage[colorlinks=true, allcolors=black]{hyperref}
\usepackage{tabularx,booktabs}
\newcolumntype{C}{>{\centering\arraybackslash}X} 
\newtheorem{lemma}{Lemma}

\newtheorem{remark}{Remark}

\definecolor{green}{RGB}{11,155,13}

\SetCommentSty{mycommfont}
\newcommand{\longthmtitle}[1]{\mbox{} \emph{(#1):}}
\SetKwRepeat{Do}{do}{while}

\usepackage{yaacro}

\begin{acgroupdef}[list=acronimos]
    \acdef{JSG}{Joint State Graph}
\end{acgroupdef}

\newcommand{\nodeset}{\mathcal{V}}
\newcommand{\edgeset}{\mathcal{E}}
\newcommand{\edge}[2]{e_{#1,#2}} %edge on LL
\newcommand{\state}[2]{s_{#1 #2}} %node in JSG
\newcommand{\edges}[2]{e_{#1,#2}}
\newcommand{\supnodes}[2]{\mathcal{Z}_{#1,#2}}
\newcommand{\cost}[2]{c_{#1,#2}} %constant edge cost
\newcommand{\costn}[2]{c^{#1}_{#2}} %cost of agent n
\newcommand{\supcost}{\tilde{c}} %cost of supporint
\newcommand{\supmove}[2]{\tilde{c}_{#1,#2}} %cost of being supported when moving
\newcommand{\Cost}[2]{C^{#1}_{#2}} %sum of cost at edge on LL graph
\newcommand{\post}[2]{p^{#1}_{#2}}
\newcommand{\actset}[2]{\mathcal{A}^{#1}_{#2}}
\newcommand{\act}[3]{a^{#1}_{#2 #3}}
\newcommand{\pathset}[1]{\mathcal{R}^{#1}}
\newcommand{\paths}[2]{r^{#1}_{#2}}
\newcommand{\CostJ}[2]{C_{#1,#2}} %edge cost on JSG
\newcommand{\pathsetJ}{\mathcal{U}} %set of paths on JSG
\newcommand{\pathJ}[1]{u_{#1}} %a path in JSG

\newcommand{\bgraph}{base graph}
\newcommand{\egraph}{environment graph}

\title{\LARGE \bf
Team Coordination on Graphs with State-Dependent Edge Cost
}

\author{Sara Oughourli$^*$, Manshi Limbu$^*$, Zechen Hu$^*$, Xuan Wang, Xuesu Xiao, and Daigo Shishika
\thanks{Sara Oughourli and Daigo Shishika are with the Mechanical Engineering Department at George Mason University. Emails: {\tt\small \{soughour,dshishik\}@gmu.edu}}
\thanks{Manshi Limbu and Xuesu Xiao are with the Computer Science Department at George Mason Univeristy. Emails: {\tt\small \{klimbu2,xiao\}@gmu.edu}}
\thanks{Zechen Hu and Xuan Wang are with the Electrical Engineering Department at George Mason Univeristy. Emails: {\tt\small \{zhu3,xwang64\}@gmu.edu}}
\thanks{*The authors contributed equally as co-first authors.}
}

\begin{document}

\maketitle
\thispagestyle{empty}
\pagestyle{empty}

%%%%%%%%%%%%%%%%%%%%%%%%%%%%%%%%%%%%%%%%%%%%%%%%%%%%%%%%%%%%%%%%%%%%%%%%%%%%%%%%

\begin{abstract}
    %\xuesu{jumped too quickly into the technical details here, would recommend to start with the big picture here. }
    This paper studies a team coordination problem in a graph environment. 
    Specifically, we incorporate ``support'' action which an agent can take to reduce the cost for its teammate to traverse some edges that have higher costs otherwise.
    Due to this added feature, the graph traversal is no longer a standard multi-agent path planning problem.
    % -- rather it must be treated as a version of MDP.
    % We provide an algorithm for converting a graph environment into a Joint State Graph (JSG) which transforms a multi-agent coordination system into a single-agent system in order to solve its path planning problem. 
    % We define the problem formulation as a base graph with a number of nodes, edges, and constant edge costs. 
    % Instead of solving an MDP, 
    To solve this new problem, we propose a novel formulation that poses it as a planning problem in the joint state space: the \emph{joint state graph} (JSG).
    Since the edges of JSG implicitly incorporate the support actions taken by the agents, we are able to now optimize the joint actions by solving a standard single-agent path planning problem in JSG.
    %We then add to this definition the action set that each agent can take and the corresponding cost for each agent to traverse an edge on the graph. We define the JSG and provide and algorithm for constructing it. 
    One main drawback of this approach is the curse of dimensionality in both the number of agents and the size of the graph.
    To improve scalability in graph size, we further propose a hierarchical decomposition method to perform path planning in two levels.
    % We then simplify the problem by decomposing the environment graph into the base and support graphs, which allows us to construct a Critical Joint State Graph (CJSG). 
    We provide complexity analysis as well as a statistical analysis to demonstrate the efficiency of our algorithm.
    % \sara{once we have the numerical results, I can finalize the abstract and conclusion sections}
    
\end{abstract}

%%%%%%%%%%%%%%%%%%%%%%%%%%%%%%%%%%%%%%%%%%%%%%%%%%%%%%%%%%%%%%%%%%%%%%%%%%%%%%%%
\section{INTRODUCTION}

 % The study of multi-agent systems is continuously being explored, as this field offers a wide range of applications for both human-robot  and robot-robot cooperation. Interest in this research can be found in the manufacturing industry, warehouse management, intelligent transportation, and many more. 
 In this work, we are interested in designing coordinated group motion, where the safety or cost for one agent to move from one location to another may depend on the support provided by its teammate. As an example, let's say there are two robots traversing an environment represented as a graph in Fig. \ref{fig: intro-example}.
 Starting from 1, the robots face a wall, represented by a red edge. 
 The robots could either climb a ladder together and potentially fall and break (move from 1 to 4 together), or one robot could hold the ladder (support from 2) while the other moves up from 1 to 4. The former option is high risk, while the latter is low risk and preferable. 
 Alternatively, if the ladder is bolted to the ground, then climbing together can be low risk and preferable.
 This paper develops a framework to study when such coordination is beneficial.

The terms cooperation and coordination take various meanings in different contexts. 
There is research done on the coordination of actions of agents to reach a state of order, such as consensus and formation control \cite{chandler2001uav, ren2005consensus, parker1993designing,lavretsky2002f, olfati2006flocking}. 
Others study cooperation in terms of simultaneously performing tasks in a spatially extended manner, like in surveillance \cite{chandler2001uav, liu2021team} and sampling \cite{bellingham2006autonomous}. 
Cooperation is also explored in problems where agents need to react locally to avoid conflict or collision, as can be seen in transportation systems on the road \cite{path2006california}, in the air \cite{tomlin1998conflict}, and in general robotic cooperation problems \cite{liu2019task,kvarnstrom2011planning, hart2020using}. 
%Many of the the types of problems mentioned so far are addressed in a distributed system, whereas we are interested in a more centralized approach.
%In task and path planning and multi-agent assignment problems \cite{liu2019task,kvarnstrom2011planning}, an agent only has to perform a set of tasks and avoid colliding with other agents. 
% In these types of problems, the coordination between the agents ends once the assignments are allocated and the paths are defined. 
We see in these situations that there is little coupling between the agents -- agents do not rely on each other to make progress, but simply need to not be in each other's paths. 
In this work, we are interested in tightly coupled agents that depend on each other for \emph{support} in order to meet their objective.
% A more sophisticated type of cooperation involves a group of agents and human operators that work together to perform a task \cite{karami2010human}.

%, monitoring \cite{traub2007terrestrial}, among others

%There is research done on the concept of consensus, such as synchronization or rendezvous, where the agents reach a common agreement to reach a certain state (location, time, etc.) \cite{chandler2001uav, ren2005consensus}.
%Another type of cooperation is formation, where agents coordinate to form predefined geometrical configuration \cite{parker1993designing,lavretsky2002f, olfati2006flocking,shishika2020game}.
%In surveillance, cooperative classification involves agents to work jointly to maximize the probability of classifying a target correctly \cite{chandler2001uav}.
%A mixed initiative cooperative problem involves a group of agents and human operators that work together to perform a task \cite{karami2010human}.
%Research in cooperation in a network of sensors that are positioned to maximize the amount of information gathered, such as in environmental sampling \cite{bellingham2006autonomous} and distributed aperture observing \cite{traub2007terrestrial}.
%There is also cooperation being studied for transportation systems for collision management on the road \cite{path2006california} and in the air \cite{tomlin1998conflict}.

\begin{figure}[t]
    \centering
    \includegraphics[width=0.47\textwidth]{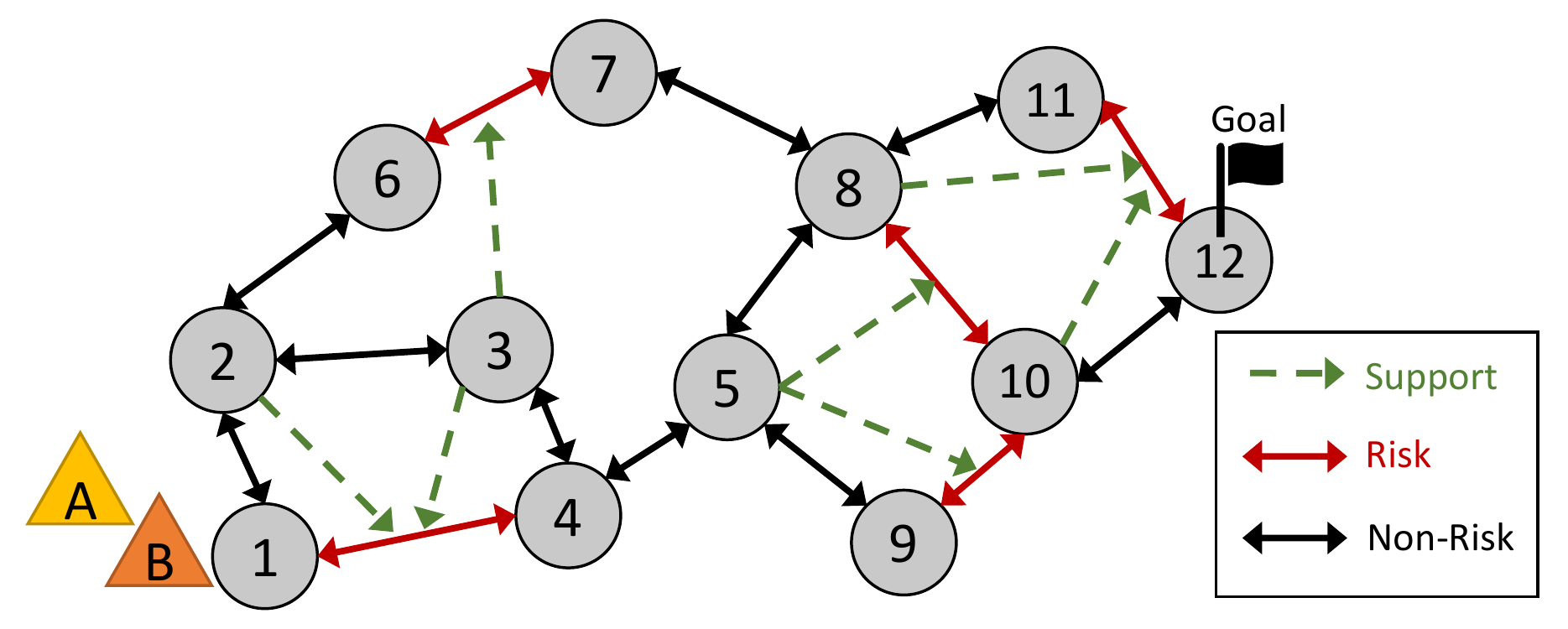}
    \caption{Example of an environment graph with risk edges and supporting nodes.} 
    \label{fig: intro-example}
\end{figure}

We study support in the context of mitigating some risks that exist in the environment. Such risk has been formulated and studied in various ways.
%\st{Different researchers define “risk” differently} based on the context of the problem they are trying to solve. 
For instance, probability of achieving certain levels of performance in a stochastic setting has been considered \cite{bertsekas2007dynamic, lim2012stochastic, xiao2019explicit, xiao2020robot}.
%For instance, \st{much of the research done in planning problems consider} stochasticity and consider the probability of achieving a certain level of performance \cite{bertsekas2007dynamic},\cite{lim2012stochastic}. 
Others have considered types of risk measures such as coherent risk measures \cite{artzner1999coherent}, like conditional value-at-risk (CVaR) \cite{ahmadi2021risk},\cite{ahmadi2020risk} and entropic value-at-risk (EVAR) \cite{ahmadi2012entropic}. Risk can also be characterized in terms of chance constraints \cite{yang2020chance}. Game theory is considered to account for the risk associated with the uncertainty in the adversary’s behavior \cite{shishika2020game}. Yet, risk can purely be described as the “cost” of traversal \cite{bertsekas2007dynamic}. In this work, we will only use this cost of traversal approach to simplify the analysis.

Cooperation has been studied both in centralized and distributed settings.
% Another element researchers in the field of cooperation consider is whether to use a centralized or a decentralized, or distributed, system. 
%Distributed systems are better when dealing with the coordination of mobile robots in a dynamic environment where there’s uncertainty \cite{le1990combination}, as agents can renegotiate parts of the plan during execution and do not need to fully collaborate \cite{kvarnstrom2011planning}. 
%On the other hand, a centralized system provides the knowledge of task interactions. 
%In addition, it allows the decomposition of the scheduling and planning given the characteristics of the problem, which improves both the efficiency and quality of execution of the planning process \cite{le1990combination}. Also, a centralized system simplifies the generation of high-quality plans and permits a centralized authority to modify a plan prior to execution \cite{kvarnstrom2011planning}. A centralized approach is better suited for tightly coupled agents that require a high degree of coordination \cite{luna2011push}. For that reason, we use a centralized approach in our work.
Decentralized systems are better at handling scalability and computational efficieny \cite{bhattacharya2010multi, wu2011online}. When it comes to Distributed Continual Planning (DCP) \cite{durfee1999survey}, plan generation and execution can happen concurrently. As it relies on communication between agents, it is better suited for online planning. On the other hand, centralized systems are better for offline planning \cite{loizou2002closed}. It is less likely to suffer from communication costs, information loss, and synchronization issues \cite{khonji2022multi}. A centralized approach is better suited for tightly coupled agents that require a high degree of coordination \cite{luna2011push}. For that reason, we use a centralized approach in our work.

Since we take a centralized approach, ensuring computational tractability becomes a challenge.
Approaches to simplifying a multi-agent planning problem have been widely studied, such as decomposition, graph reformulation, and others \cite{erol1994htn,yu2013multi,luna2011push}. 
% Plan decomposition methods, such as hierarchical task network (HTN) planning \cite{erol1994htn}, focus on a distributed system of agents with a possibly incomplete picture of other agents’ plans. Making the connection between multi-agent path planning collision-free unit-distance (CUG) and network flow \cite{yu2013multi} has shown promising results in terms of time efficiency and optimality in path planning. An algorithm for high levels of coordination to avoid collision \cite{luna2011push} has proven to be fast and provide completeness guarantees for graphs without making any assumptions on their topology. 
In our work, we develop a hierarchical decomposition method on a reformulated graph to solve a multi-agent path planning problem with high coordination.

%Lim and Tsiotras  employ a search algorithm where each agent builds its own search space of the environment, communicates with the other agents in the same environment, and updates its own search. This results in all agents acquiring a common subgraph that includes an optimal path. \cite{lim2020mams} 

%Luna and Berkis  develop an algorithm that employs a “push” operation (agents move towards their goal until no progress can be made) and “swap” operation (agents swap positions without changing the configuration of other agents). This algorithm works well for problems requiring high levels of coordination between hundreds of agents. \cite{luna2011push}

%To determine if cooperation is required, Zhang and Subbarao \cite{zhang2016formal} address two questions: what conditions need to be met for cooperation to be required, and how to determine the minimum number of agents required to solve the problem. They determined that, if none of these conditions hold, then the problem is single-agent solvable.

The contribution of the paper are: (i) the formulation of a new multi-agent coordination problem with strong coupling between teammembers' positions and action; (ii) a conversion of the problem into a simple single-agent path-planning problem; and (iii) development of a hierarchical decomposition scheme that alleviates the curse of dimensionality.

\section{PROBLEM FORMULATION}
\label{sec:problem-formulation}

 We consider a scenario where a team of robots must move from their initial locations to some goal locations.
More specifically, we are interested in a situation where the cost of traversal is affected by the presence and actions of other team members.
In the following, we will introduce the base graph, and then formulate how the edge cost changes based on the ``support'' provided by the teammate.
For conciseness, we will restrict the discussion to a two-agent team, but the idea will generalize to a larger team size.\footnote{The computational complexity will be an important consideration for scalability.} 
% Describe the problem we are trying to solve.\\

The environment is modeled as a graph where nodes represent key locations and the edges represent the traversability between them.
%\daigo{Figure 1 should really be discussed with the \egraph{} right?}
The \bgraph{} is denoted by \(\mathbb{G}=(\nodeset,\edgeset)\), where \(\nodeset\) is a set of nodes, and \(\edgeset\) is a set of edges, \(\edgeset\subset \nodeset\times \nodeset\). 
We assume $\mathbb{G}$ is strongly connected.
The starting positions of the agents are denoted by the node set $\nodeset_0 \subset \nodeset$.
The robots seek to reach a set of goal nodes $\nodeset_g\subset \nodeset$ while minimizing the cost of traversal.
%We have a multi-agent system with a set of starting nodes $V_1 \subset \nodeset$, where each robot has a different starting node. 
%The set of goal nodes are also known and given as $V_g \subset \nodeset$. 
The nominal cost for traversing the edge \(\edge{i}{j}\in \edgeset\) is a given constant, $\cost{i}{j}$ for $i,j\in \nodeset$. 
% \daigo{We can define the distance / cost of traversal $\psi$ (used in Lemma 1) here.}
% \xuan{If we define the notation of path here, we can introduce $\psi$ of that path. }
%We define the cost of traveling a path $I$ as $\psi_I$. 

Let $I_{ab}$ denote a path (set of edges) from $a\in \nodeset$ to $b \in \nodeset$.
We use $\psi_{a,b}$ to denote the minimum cost to move from $a$ to $b$:
\begin{equation}
    \psi_{a,b} = \min_{I_{ab}} \sum_{e_{i,j}\in I_{ab}} c_{i,j}.
\end{equation}
A standard path planning will simply consider this shortest path for each agent.

%\xuan{By `minimum' cost, we need to define `in which graph'. In Section III B, we have an environment graph, an ordinary graph, and a supporting graph. There, the $\psi$ is defined as the minimum cost from $a$ to $b$ on the \textbf{ordinary graph}. If there are multiple places we need notation $\psi$, we can define it here. If it is only used in Section IIIB, defining it there is clearer.}
%\sara{we do address minimum costs on each graph, so I think keeping this definition here}

We now define the \emph{\egraph{}}, which incorporates the notion of risk and support. Each edge $\edge{i}{j}$ is associated with a set of support nodes, $\supnodes{i}{j} \subseteq \nodeset$. If this set is non-empty, then an agent at $v\in\supnodes{i}{j}$ can provide support for the agent traversing $\edge{i}{j}$. 
The action set for an agent $n\in \{A,B\}$ at node $i$ is given as
$\actset{n}{i}=\{\{\act{}{i,}{j}\}_{j\in\mathcal{N}_i}, \act{}{}{s}\}$. 
Where $\mathcal{N}_i$ is the neighborhood of $i$, and $\act{}{i,}{j}$ is the action to move to node $j$ given that it is in the neighborhood of $i$. The action $\act{}{}{s}$ is the support. Note, if we were to consider a team size larger than two, we would have to explicitly denote which agent is being supported by $n$.

Let $\post{t}{} = (\post{t}{A},\post{t}{B})$ be the position of agents A and B at time $t$, and let $\act{t}{}{} = (\act{t}{A}{},\act{t}{B}{})$ be the actions agents A and B take at time $t$. The cost of an action for agent A is given as 
\begin{equation}
\costn{t}{A}(\cdot)=
\begin{cases}
    \cost{i}{j},& \text{if } \act{}{A}{} = \act{}{i,}{j} \text { and } \post{}{B} \notin \supnodes{i}{j} \text{ or } \act{}{B}{} \neq \act{}{s}{}, \\
    \supmove{i}{j},& \text{if } \act{}{A}{} = \act{}{i,}{j} \text{, }\post{}{B} \in \supnodes{i}{j} \text{, and } \act{}{B}{}=\act{}{s}{}, \\
    \supcost, & \text{if } \act{}{A}{}=\act{}{s}{}, \\
    0, &\text{if } \act{}{A}{} \neq \act{}{s}{} \text{ and } \act{}{A}{} \neq \act{}{i,}{j},
\end{cases}
\end{equation}
where $(\cdot)$ represents the arguments $(\post{t}{},\act{t}{}{},\supnodes{i}{j})$.

% \daigo{Maybe explain this with Fig.~1, discussing different scenarios like A and B both moving vs. A staying and B moving.}\sara{see example below}

\begin{figure}[t]
    \centering
    \includegraphics[width=0.40\textwidth]{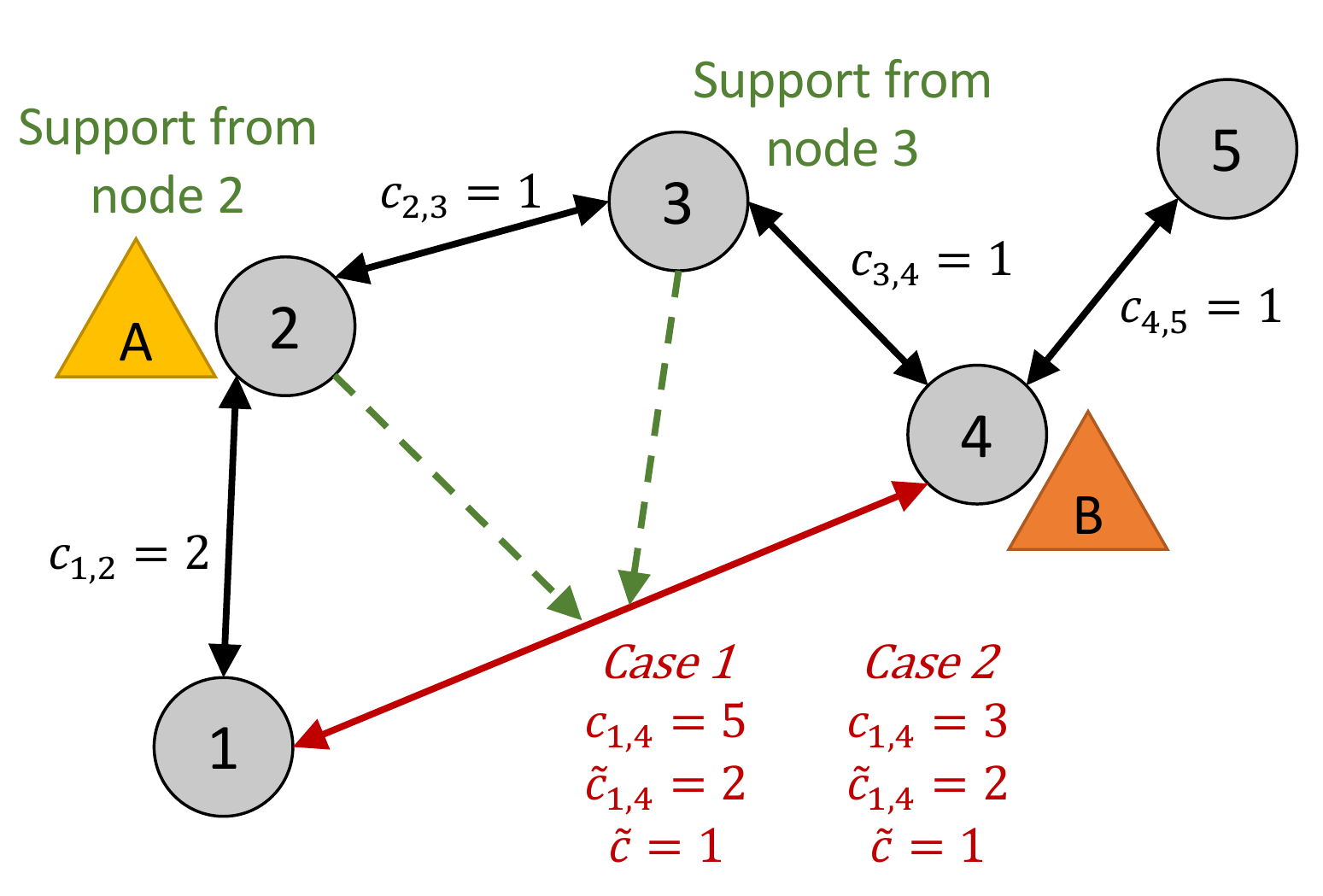}
    \caption{Illustrative example of an environment graph with a risk edge and supporting nodes. Case 1 has a high risk cost. Case 2 has a low risk cost.} 
    \label{fig: example}
\end{figure}

For example, in Fig. \ref{fig: example}, if at $t=1$ agent A is at node 2 (a supporting node) and provides support to agent B as the latter moves from node 1 to node 4, then the cost for agent A would be $\costn{1}{A}=\supcost$ and the cost for agent B would be $\costn{1}{B}=\supmove{1}{4}$. If both agents A and B move together from node 1 to node 4, the cost for the agents would be $\costn{1}{A}=\costn{1}{B}=\cost{1}{4}$.

In order to find the total cost at time $t$, we can simply sum the costs for both agents

\begin{equation}
\Cost{t}{}=\costn{t}{A}+\costn{t}{B}.
\end{equation}

Let $\pathset{n} = \{\paths{n}{1},...,\paths{n}{m}\}$ be the set of action sequences agent $n$ can take from start node to goal node. Where each sequence is the ordered set of actions taken from start to goal from $t=1$ to the time it takes for the agents to reach the goal state, $T$, i.e., $\paths{n}{m}=[\act{1}{n}{},...\act{T}{n}{}]$. 

%\marginXW{Do we only consider two agents? Yes.}

The sum of the costs for a given sequence of actions $\paths{A}{}\in\pathset{A}$, $\paths{B}{}\in\pathset{B}$ are
\begin{equation}
F(\paths{A}{},\paths{B}{})= \sum_{t=1}^{T} \Cost{t}{}. 
\end{equation}
The goal is to find a pair ($\paths{A*}{},\paths{B*}{}$) that minimizes the total cost, $F$:

\begin{equation} \label{5}
\min_{\paths{A}{}\in \pathset{A},\paths{B}{}\in \pathset{B}} F.
\end{equation}

% \daigo{The section should not end abruptly with an equation. In the least you can say what is coming next.}  \daigo{We can also mention different ways in which the pair can move to the goal with Fig.1 in mind, and lead to the following sections actually finding the optimal.} \daigo{Also, to tell the full story, it might be good to specifically mention where the goal is in Fig.1.} \sara{see this paragraph I've written}

%If agents A and B in Fig. \ref{fig: example} start from node 1 and want to reach the goal, node 5, they must decide whether they support each other in doing so. For example, they could traverse the edge $\edge{1}{4}$ while one supports the other, leading to a reduced total cost, and from there go to node 5. Or, they could take the longer route: 1, 2, 3, 4, then 5. Although they avoid the risk edge, they might possibly accrue more costs in doing so.
%In the following section, we introduce a method of transforming the environment graph to a graph of the joint state space to perform path planning on.

An illustrative example in Fig. \ref{fig: example}, where agents A and B need to reach goal node 5. To traverse the risk edge, they either use or do not use support depending on how costly, or risky, the edge is.
Agents demonstrate supporting behavior in Case 1.
If agent B traverses risk edge $\edges{1}{4}$ without support from A, the cost for B would be $\cost{1}{4}=5$.
With support from A, the reduced cost for B would be $\supmove{1}{4}=2$. The total cost at this time step $t$ is $\Cost{t}{}=\supmove{1}{4}+\supcost=2+1=3$. 
Thus, agents in this case accrue less costs by supporting each other.
Case 2 is a scenario where the agents do not show supporting behavior in a low risk situation. 
Since $\cost{1}{4}=3$, B can traverse $\edges{1}{4}$ without support from A. The total cost at this time step $t$ is $\Cost{t}{}=\cost{1}{4}+0=3+0=3$ without A's support.

One way to solve the minimization problem in (\ref{5}) is by posing it as an instance of MDP. However, we will introduce a simplification using the concept of Joint State Graph in the next section.

\section{METHOD}
\label{sec:method}

%\daigo{Briefly mention what this section is going to do.}\sara{please check what I wrote}

 We first introduce the concept of Joint State Graph (JSG) which simplifies the joint action selection problem into a standard path planning problem on graphs.  To improve scalability, we propose in \ref{Sec_CJSG} a method of decomposing JSG to deal with scalability of the graph. Finally, we conclude the section with a complexity analysis.

\subsection{Joint State Graph}\label{Sec_JSG}

 The problem described in the previous section can be solved using MDP. However, we propose transforming the environment graph into a joint state space graph. The paths on the JSG inherit the actions of the agents, which means we no longer need to consider the action sets. This makes the problem simpler than MDP. Let the JSG be a graph \(\mathbb{J}=(\mathcal{S},\mathcal{L})\), where $\mathcal{S}=\{\state{i}{j}:i,j\in \nodeset\}$ is the set of nodes representing joint states, and \(\mathcal{L}\) is the set of edges. 

Let $\state{1}{1}$ be the initial state assuming that $\nodeset_0=(1,1)$. Let $\state{g}{g}$ be the goal state. 
Edges on JSG are denoted as $\edges{ij}{wk}=(\state{i}{j},\state{w}{k})$ if agent A can move from \(i\) to \(w\) in the \bgraph{}, and agent B can move from \(j\) to \(k\), i.e., \(\edge{i}{w}\in \edgeset\) and \(\edge{j}{k}\in \edgeset\). If agent A does not move, we have $i=w$. Similarly, if B does not move, $j=k$.

%$\edges{ij}{ik}=(\state{i}{j},\state{i}{k})$ given that agent B can move from node \(j\) to node \(k\) in the \bgraph{}, i.e., \(\edge{j}{k}\in \edgeset\). And $\edges{ij}{wj}=(\state{i}{j},\state{w}{j})$ is an edge if agent A can move from node \(i\) to node \(w\) in the \bgraph{}, i.e., \(\edge{i}{w}\in \edgeset\). And $\edges{ij}{wk}=(\state{i}{j},\state{w}{k})$ is an edge if agent A can move from node \(i\) to node \(w\) in the \bgraph{} and agent B can move from node \(j\) to node \(k\), i.e., \(\edge{i}{w}\in \edgeset\) and \(\edge{j}{k}\in \edgeset\). Fig. \ref{fig: jsg-example} is an example of a JSG using the 5-node environment graph in Fig. \ref{fig: example}.

Let $\mathcal{C}$ be the set of costs for each edge on the JSG, where an element is denoted as $\CostJ{ij}{wk}$. 
%If agents at $\state{i}{j}$ move to $\state{w}{k}$, where $i\neq w$ and $j \neq k$, this indicates that both agents took the action to move instead of staying in place and support the other traverse an edge.  
%This lack of supporting behaviour makes the cost of the edge $\edges{ij}{wk}$ be $\CostJ{ij}{wk} = \cost{i}{w} + \cost{j}{k}$. 
If A remains at $i \in \supnodes{j}{k}$ while B traverses from $j$ to $k \in \mathcal{N}_j$, the cost is defined as $\CostJ{ij}{ik}= \min \{\cost{j}{k}, (\supmove{j}{k}+\supcost)\}$. This is how the edge in JSG subsumes the action selection in the original problem. 
However, if $i \notin \supnodes{j}{k}$, then the  cost is simply $\CostJ{ij}{ik} = \cost{j}{k}$. 
The case when A moves is defined similarly.
%Similarly, when A traverses from $i$ to $w \in \mathcal{N}_i$  while B remains at $j \in \supnodes{i}{w}$, the cost becomes $\CostJ{ij}{wj} = \min \{\cost{i}{w}, (\supmove{i}{w} + \supcost$)\}. 
%And if $j \notin \supnodes{i}{w}$, then the cost is $\CostJ{ij}{wj} = \cost{i}{w}$.
This explains lines 7-17 of Algorithm \ref{Alg_defJSG}.
If A traverses from $i$ to $w \in \mathcal{N}_i$ and B traverses from $j$ to $k \in \mathcal{N}_j$, then we add the nominal costs $\CostJ{ij}{wk} = \cost{i}{w} + \cost{j}{k}$. In the case that both agents are stationary, the cost is $\CostJ{ij}{ij} = 0$
The details of the JSG construction are in Algorithm \ref{Alg_defJSG}.

%If $\post{t}{A} \in \supnodes{j}{k}$, the cost $\CostJ{ij}{ik}=\min {}$
%\daigo{This part is essentially explaining lines 7-15 of Algorithm 1. Can specifically mention that.}

Let $\pathsetJ = \{\pathJ{1},...,\pathJ{m}\}$ be the set of paths on the JSG, where an element $\pathJ{}=\{\edges{11}{kw},...,\edges{ij}{gg}\}$ is the set of edges from the initial state to the goal state. Then, the cost of each path $\pathJ{}$ is given by the sum of the cost of the edges on that path in JSG,
\begin{equation}
 Q({u}) = \sum_{\edges{ij}{wk} \in \pathJ{}} \CostJ{ij}{wk}.
\end{equation}
Using a standard shortest-path algorithm, we can find the optimal path that minimizes $Q({u})$:
% over the paths in $\pathsetJ$, gives the optimal path for JSG,
\begin{equation}
 Q({u}^{\star}) =\min_{\pathJ{} \in \pathsetJ{}}  Q({u}).
\end{equation} 

An example of JSG is shown in Fig.~\ref{fig: jsg-example},
\begin{figure}[t]
    \centering
    \includegraphics[width=0.45\textwidth]{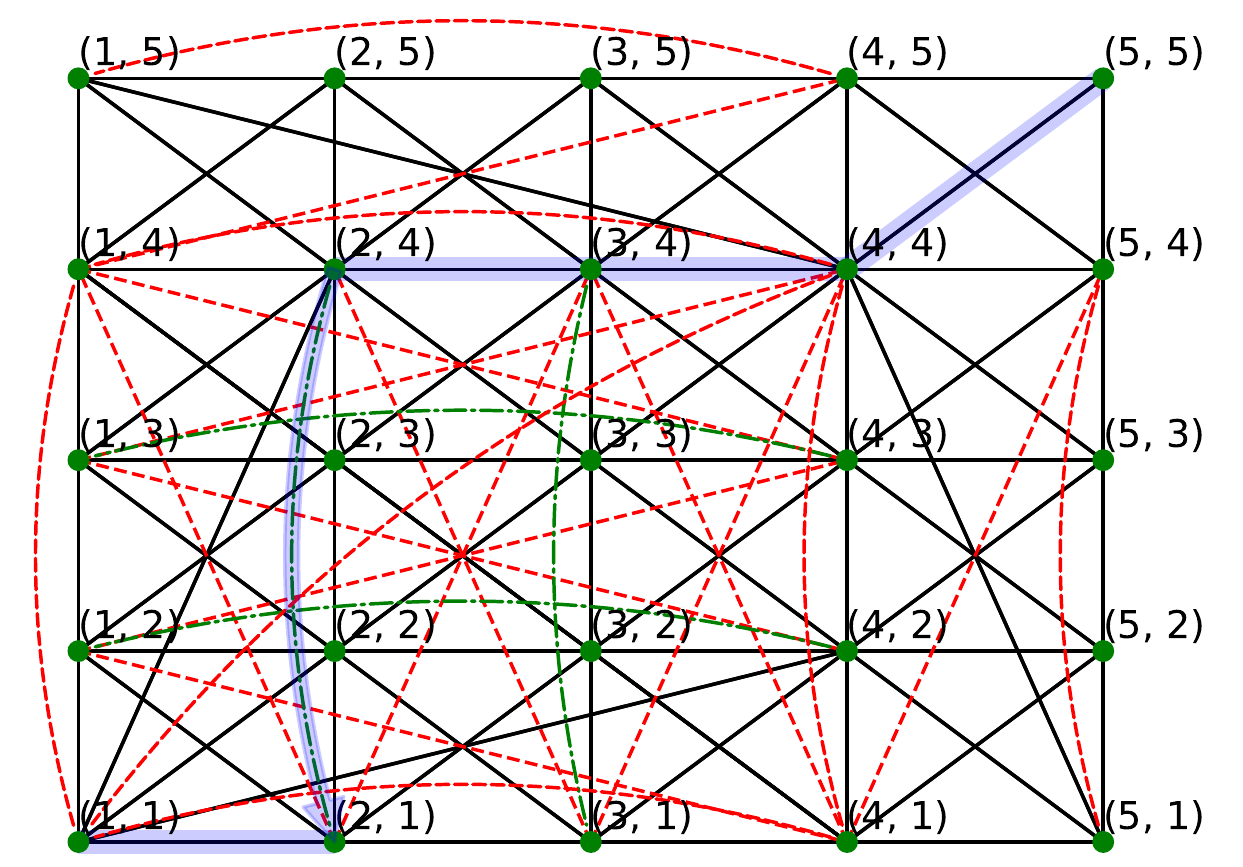}
    \caption{Code-generated Joint State Graph from the 5-node environment graph. Black edges are non-risky edges, red edges are risk with no support, and green edges are risk with support.}
    \label{fig: jsg-example}
\end{figure}
which corresponds to the environment graph shown in Fig. \ref{fig: example}.
% We now have a graph that subsumes actions, and we want to perform path planning on this graph. Fig. \ref{fig: jsg-example} is the JSG for the environment graph in Fig. \ref{fig: example}. 
The edges highlighted in blue indicates the optimal path $u^\star$ for Case~1 in Fig.~\ref{fig: example}. 
Importantly, we can easily identify the original actions from the edges selected in this JSG: e.g., the use of edge $e_{21,24}$ indicates that A at node 2 supported B who moved from node 1 to 4.
% edges form a path that implies what actions the agents took.

Although planning on JSG is conceptually simple, it can become computationally expensive with greater graph sizes. The next section addresses this issue.

\begin{algorithm2e}
    \label{Alg_defJSG}
    \caption{JSG Construction.}
    \SetAlgoLined
    \textbf{Input} $\mathbb{G}=(\mathcal{V},\mathcal{E})$, $\mathcal{Z}_{i,j}$.\\
    \textit{Let} $\mathbb{J}=(\mathcal{S}, \mathcal{L})$\\
    \For{$\forall i,j\in \mathcal{V}$}
    {
    Add $s_{ij}$ to $\mathcal{S}$
    }
    \For{any two distinct elements $s_{ij},s_{wk} \in \mathcal{S}$}
    {\uIf{$i=w$, $j\neq k$ and $k\in\mathcal{N}_{j}$ }
    {\uIf{$i\in$ $\mathcal{Z}_{jk}$ }
   {
   Add edge $\edge{ij}{wk}$ to $\mathcal{L}$. Define its cost as $C_{ij,wk}=\min\{c_{j,k}, (\supmove{j}{k}+\supcost)\}$\\
   }    
   \Else{
    Add edge $\edge{ij}{wk}$ to $\mathcal{L}$. Define its cost as $C_{ij,wk}=c_{j,k}$\\
    }
    }
    {\uElseIf{$i\neq w$, $j=k$ and $w\in\mathcal{N}_i$}
    {\uIf{$j\in \mathcal{Z}_{iw} $}
    {
    Add edge $\edge{ij}{wk}$ to $\mathcal{L}$. Define its cost as $C_{ij,wk}=\min\{c_{i,w},(\supmove{i}{w}+\supcost)\}$\\
    }    
    \Else{
    Add edge $\edge{ij}{wk}$ to $\mathcal{L}$. Define its cost as $C_{ij,wk}=c_{i,w}$\\
    }}}
    {\ElseIf{$k\in\mathcal{N}_{j}$ and $w\in\mathcal{N}_i$}{
    Add edge $\edge{ij}{wk}$ to $\mathcal{L}$. Define its cost as $C_{ij,wk}=c_{i,w}+c_{j,k}$
    }}
    }
   \textbf{Return} $\mathbb{J}=(\mathcal{S}, \mathcal{L})$ and the associated costs $C_{ij,wk}$.
\end{algorithm2e}

\subsection{Search Algorithm: Critical Joint State graph}\label{Sec_CJSG}
 In this section, we introduce a new search algorithm based on constructing a Critical Joint State Graph (CJSG), which has reduced computational complexity compared with the straightforward JSG method in Sec. \ref{Sec_JSG}. 
Note that the Joint State Graph $\mathbb{J}$ has $|S|=|\nodeset|^2$ number of nodes, leading to high complexity if directly used for planning.
To address this issue, our key idea is to classify the agents' movements into coupled and decoupled modes, where only the coupled movements need to be planned on a joint state representation, and the decoupled movements can be independently planned by each agent on \bgraph{} $\mathbb{G}$.
%We start with a decomposition of the environment graph formulated in section II. Based on whether the edges in $\mathbb{G}$ are associated with support nodes, we categorize them into a base edge set $\mathcal{E}_B$ such that $\forall e_{i,j}\in \mathcal{E}_B$, $\supnodes{i}{j}=\emptyset$, and a supported edge set $\mathcal{E}_R$ such that $\forall e_{i,j}\in \mathcal{E}_R$, $\supnodes{i}{j}\neq\emptyset$. Clearly, $\mathcal{E}_B\bigcup \mathcal{E}_R=\mathcal{E}$ and $\mathcal{E}_B\bigcap \mathcal{E}_R=\emptyset$. %The two types of edges allow us to decompose the environment graph $\mathbb{G}$ into a base graph $\mathbb{G}_B$ and a support graph $\mathbb{G}_S$, as visualized in Fig. \ref{fig_dcp}.
As visualized in Fig. \ref{fig_dcp}, the environment graph formulated in Sec. \ref{sec:problem-formulation} builds on a base graph $\mathbb{G}$, then associates some of its edges with a set ($\supnodes{i}{j}$) of support nodes. 
Depending on whether the edges in $\mathbb{G}$ have at least one support node, we define a risk edge set $\mathcal{E}_R$ such that $\forall e_{i,j}\in \mathcal{E}_R$, $\supnodes{i}{j}\neq\emptyset$. %Clearly, $\mathcal{E}_B\bigcup \mathcal{E}_R=\mathcal{E}$ and $\mathcal{E}_B\bigcap \mathcal{E}_R=\emptyset$.
%Then, as visualized in Fig. \ref{fig_dcp}, we can decompose the environment graph into a base graph $\mathbb{G}$ and a support graph. 
Note that the support graph in Fig. \ref{fig_dcp} does not follow the standard `graph' definition in mathematics. It only describes a supporting relationship between nodes and risk edges which we use later to study coupled movements of agents.
%and a supported edge set $\mathcal{E}_R$. Each supported edge $e_{ij}\in\mathcal{E}_R$ is associated with a node set $\supnodes{i}{j}$ such that for a particular time step, one agent staying at node $v_k\in\supnodes{i}{j}$ can provide support to the other agent moving across edge $e_{ij}$ or $e_{ji}$ with a cost $c_{ij,s}$.
%Let the graph with all nodes and base edges be a base graph. Let the graph with all nodes and supported edges (and the corresponding supporting nodes) be the supported graph. 
\begin{figure}[h]
    \centering
    \includegraphics[width=0.48\textwidth]{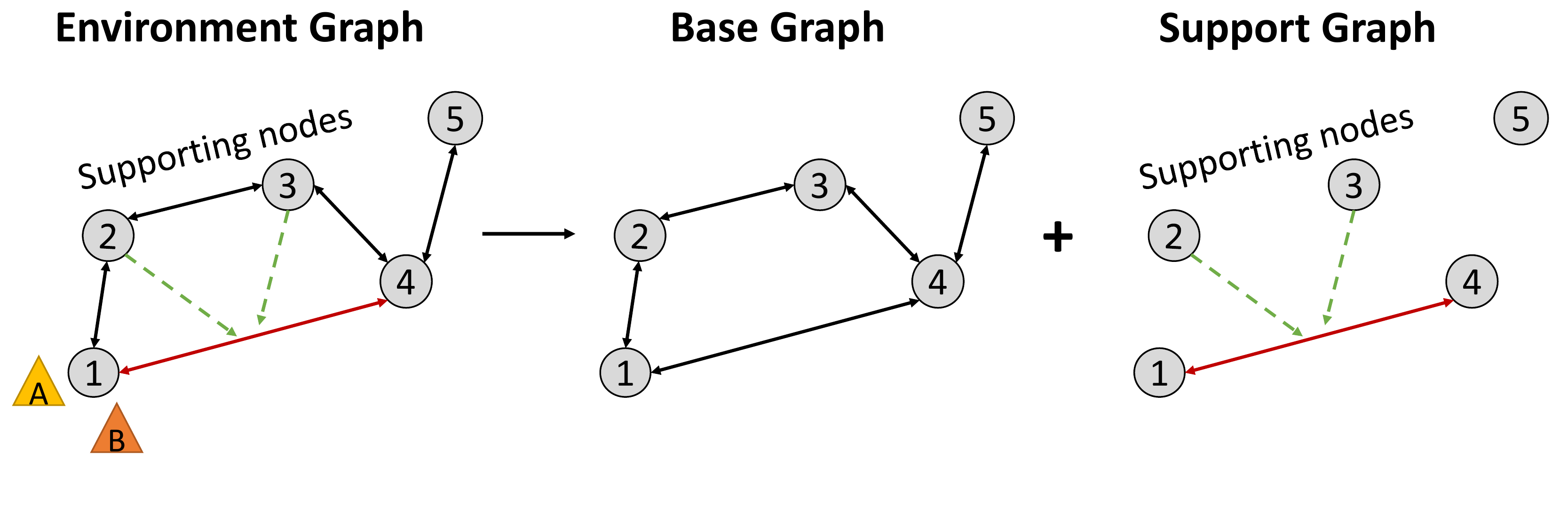}
    \caption{Example for environment graph decomposition. In $\mathbb{G}$, node 2 and node 3 can support the edge between node 1 and node 4. }
    \label{fig_dcp}
\end{figure}

We start by considering the costs for decoupled movements of the two agents. 
Recall that $\psi_{a,b}$ denotes the minimum cost for an agent to move from $a$ to $b$ on the base graph $\mathbb{G}$, the following statement holds. %, which can be easily computed by Dijkstra's algorithm with a complexity of $\mathcal{O}(|\mathcal{V}|^2)$. 

\begin{lemma}\label{Lm_base_graph} \longthmtitle{Decoupled planning on base graph}
On graph $\mathbb{G}$, consider the first agent moves from node $i$ to $w$; the second agent moves from node $j$ to $k$. Let  $R_{ij,wk}$ be the minimum cost for the two agents to complete the movement without performing supporting behaviors. Then
$$R_{ij,wk}= \psi_{i,w}+\psi_{j,k}.$$
\end{lemma}
%\daigo{$\psi$ must be defined earlier in the paper.}
 \begin{proof}
 The proof is trivial. Since the two agents do not perform supporting behaviors, their movements and associated costs can be computed individually on the base graph $\mathbb{G}$, which by definition are $\psi_{i,w}$ and $\psi_{j,k}$. 
 \end{proof}

%While Lemma \ref{Lm_base_graph} characterizes the decoupled movements, 
Now, to characterize the coupled movements of the two agents, we construct a Critical Joint State Graph (CJSG), $\mathbb{T}=(\mathcal{M}, \mathcal{H})$, where $\mathcal{M}$ and $\mathcal{H}$ are the node set and edge set of $\mathbb{T}$, respectively. For any $h_{ij,wk}\in\mathcal{H}$, let $W_{ij,wk}$ denote the cost associated with this edge.
Details of CJSG construction are summarized in Algorithm \ref{Alg_defCJSG}.  

\begin{algorithm2e}
    \label{Alg_defCJSG}
    \caption{CJSG Construction}
    \SetAlgoLined
    \textbf{Input} $\mathcal{E}_R$, $s_{11}$, $s_{gg}$, $\mathcal{Z}_{i,j},R_{ij,wk}$.\\
    \textit{Let} $\mathbb{T}=(\mathcal{M}, \mathcal{H})$\\
    \For{each $e_{i,j} \in \mathcal{E}_R$}
    {\For{each $k\in\supnodes{i}{j}$}
    {
    Add $s_{k i}$ and $s_{k j}$ to $\mathcal{M}$. \label{step_add_node}\\
    Add $s_{i k}$ and $s_{j k}$ to $\mathcal{M}$. \label{step_add_node2}\\
    }
    Add $s_{11}=(1,1)$ and $s_{gg}=(g,g)$ to $\mathcal{M}$ (if they are not already in $\mathcal{M}$).\label{step_add_node3}\\
    \For{any two distinct elements
    $s_{ij}, s_{wk}\in\mathcal{M}$}
    {\uIf{
    $e_{j,k}\in\mathcal{E}_R$ and $i=w\in \supnodes{j}{k}$ 
    }
    {Add edge $h_{ij,wk}$ to $\mathcal{H}$. Define its cost as $W_{ij,wk}= \min\{(\supmove{j}{k}+\supcost),R_{ij,wk}\}$\label{step_add_cost}\\
    }
    \uElseIf{$e_{i,w}\in\mathcal{E}_R$ and $j=k\in \supnodes{i}{w}$}
    {Add edge $h_{ij,wk}$ to $\mathcal{H}$. Define its cost as $W_{ij,wk}= \min\{(\supmove{i}{w}+\supcost),R_{ij,wk}\}$. \label{step_add_cost2}
    }
    \Else{
        Add edge $h_{ij,wk}$ to $\mathcal{H}$. Define the associated cost as $W_{ij, wk}=R_{ij,wk}$. \label{step_add_cost3}
    }
    }
    }
    \textbf{Return} $\mathbb{T}=(\mathcal{M}, \mathcal{H})$ and the associated costs $W_{ij,wk}$.
\end{algorithm2e}

\begin{remark}\longthmtitle{Algorithm \ref{Alg_defCJSG} explained}  \label{RM_Alg_defCJSG}
In CJSG, we consider the node of the graph as any joint state that the two agents (i) can initiate or complete supporting behaviors (c.f. steps \ref{step_add_node} and \ref{step_add_node2}), (ii) at their start or goal position of the planning task (c.f. step \ref{step_add_node3}). We let CJSG be fully connected. The edge costs are associated with two agents moving over the base graph (c.f. step \ref{step_add_cost3}) or a possible lower cost when they perform a support behavior (c.f. steps \ref{step_add_cost} or \ref{step_add_cost2}, depending on who supports who).
%For steps \ref{step_add_node} and \ref{step_add_node2}: one agent provides support at $v_k$, then the other agent may move between $i$ and $j$ with reduced cost. 
%$\mathbb{T}$ is a fully connected graph. 
\end{remark}

To provide a toy example, given the environment graph in Fig. \ref{fig_dcp}, we can construct the corresponding CJSG as shown in Fig. \ref{fig_CJSG}. Using the support graph, the \emph{critical joint states} are the highlighted states in the middle part of Fig. \ref{fig_CJSG} as well as the initial and goal states. By computing edge costs according to Lemma \ref{Lm_base_graph} and Algorithm \ref{Alg_defCJSG}, the CJSG is constructed as shown on the right side of Fig. \ref{fig_CJSG}. Red edges are edges under supporting behaviors such that $\supmove{j}{k}$ and $\supmove{i}{w}$ are available. The blue edges are associated with $R_{ij,wk}$ where two agents are decoupled and individually seek optimal paths on the base graph.

\begin{figure*}[htbp]
    \centering
    \includegraphics[width=0.95\textwidth]{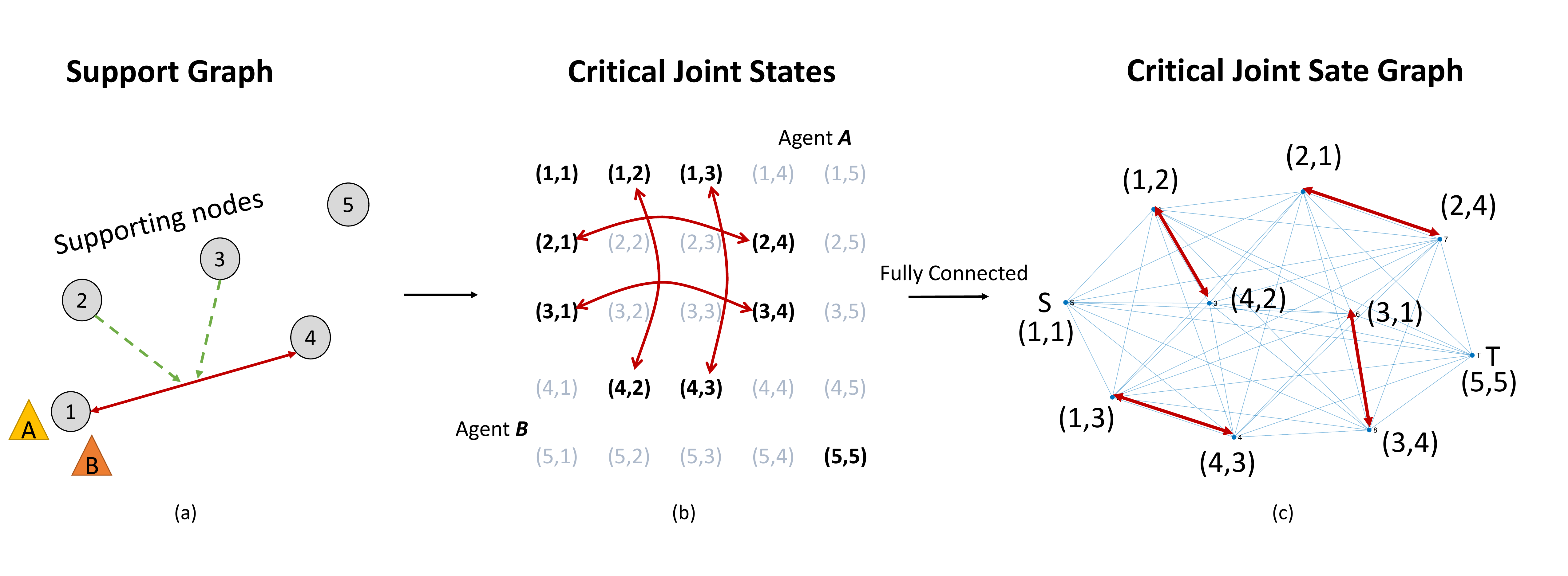}
    \caption{Example for CJSG Construction. Based on support graph in (a), we can observe all critical joint states as depicted in (b). By fully connecting these critical joint states with the initial states and goal states, we obtain the CJSG $\mathbb{T}$ as shown in (c).}
    \label{fig_CJSG}
\end{figure*}

After constructing CJSG, we present our search algorithm. We define a path composition operation. Suppose $u_1=\{e_{a,b},e_{c,d},\cdots,e_{i,j}\}$, $u_2=\{e_{g,h},e_{r,t},\cdots,e_{k,\ell}\}$. Then $u_1\oplus u_2 =\{e_{ag,bh},e_{cr,dt},\cdots,e_{il,j\ell}\}$. When the two paths do not have the same length, we extend the shorter one by repeating its final node (so that in the graph representation, it stays at that node). For example, if $u_1$ is two elements shorter than $u_2$ then we extend it by $u_1=\{e_{a,b},e_{c,d},\cdots,e_{i,j},e_{j,j},e_{j,j}\}$. The length of the composed path equals the length of the $u_1$ or $u_2$, whichever is longer.
Based on CJSG and path composition, our search algorithm is presented in Algorithm \ref{Algorithm_search_CJSG}, where \textit{PathPL} can be any path planning algorithm, i.e., Dijkstra's algorithm \cite{CTH-LCE-RRL-SC:22}, that can obtain a shortest path between two nodes.

\begin{algorithm2e}
	\label{Algorithm_search_CJSG}
	\caption{Path Planning based on CJSG.}
	\SetAlgoLined
	\textbf{Input}  $\mathbb{T}=(\mathcal{M}, \mathcal{H})$, $\mathbb{G}=(\nodeset,\edgeset)$.\\
    %\textbf{Function} Dijkstra(Graph, Start, Goal):\\
    %\STATE \textit{return} optimal path $u$.\\
    $\widehat{u} \leftarrow$ PathPL($\mathbb{T}$, $s_{11}$, $s_{gg}$) \\
    \For{each ${h}_{ij,wk}$ in $\widehat{u}$}
    {
    \uIf{
    $W_{ij,wk}=R_{ij,wk}$
    }
    {Add [PathPL$(\mathbb{G},i,w)$$\oplus$PathPL$(\mathbb{G},i,k)$)]  to ${u}^{\dagger}$ \label{step_recons}} 
    \ElseIf{
    $W_{ij,wk}= \supmove{i}{w}+\supcost$ or $W_{ij,wk}= \supmove{j}{k}+\supcost$
    }
    {Add $[e_{ij,wk}]$ to ${u}^{\dagger}$}
    }
	\textbf{Return} ${u}^{\dagger}$.
\end{algorithm2e}

\begin{remark}\longthmtitle{Algorithm \ref{Algorithm_search_CJSG} explained}  \label{RM_Alg_searchCJSG}
We first perform path planning on $\mathbb{T}$. Note that some edges ${h}_{ij,wk}$, when $W_{ij,wk}=R_{ij,wk}$, are associated with paths of two agents planned in $\mathbb{G}$ using Lemma \ref{Lm_base_graph}. We use step \ref{step_recons} to reconstruct these edges back to paths that the agents can traverse on the environment. Furthermore, although step \ref{step_recons} recalls a path planning process, this planning should have already been computed by Lemma \ref{Lm_base_graph} when executing Algorithm \ref{Alg_defCJSG}.
\end{remark}

\begin{lemma} \label{Lm_CJSG_Veri} \longthmtitle{Effectiveness of the critical-joint state graph} The following statements hold.\\
(i) Any path ${u}^{\dagger}$ reconstructed from a path $\widehat{u}$ on CJSG is a feasible path for the two agents on the environment graph, thus, $Q({u}^{\dagger}) \ge Q({u}^{\star})$. \\
(ii) The optimal path planned from the CJSG has the same minimum cost as the optimal path planned directly from the JSG, thus, $Q({u}^{\dagger}) \le Q({u}^{\star})$.
\end{lemma}

 \begin{proof}
 We prove the two statements in Lemma \ref{Lm_CJSG_Veri}. \\
(i) Given any edge $e_{ij,wk}$ in path ${u}^{\dagger}$, there are two situations. One is $W_{ij,wk}= \supmove{i}{w}+\supcost$ or $W_{ij,wk}= \supmove{j}{k}+\supcost$, and another is $W_{ij,wk}=R_{ij,wk}$. For the first case, the path in the environment graph represents one agent staying at the support node where another agent traverses the corresponding supported edge. It is a feasible path for two agents on the environment graph by definition. For the second case, according to Lemma 1, two agents move independently and decoupled. As each agent moves from one node to another on the environment graph, the path obtained from the algorithm is feasible since the planned base graph $\mathbb{G}$ is a subgraph of the environment graph. Therefore, $u^{\dagger}$ is always a feasible path for two agents on the environment graph. Since $u^{\star}$ is the optimal path for the environment graph planned from the JSG, we have $Q({u}^{\dagger}) \ge Q({u}^{\star})$. \\
(ii) We prove this by showing that for any optimal path planned from the JSG, there is a path on CJSG with the same cost. Considering the optimal path planned from JSG, it consists of two agents' movements with and (possibly) without supporting behaviors. Let the $u^{\star}$ be divided into different segments according to the above attribute. It is obvious that all such segments are connected by a joint state that initiates the supporting behavior and a state that completes the supporting behavior, which are essentially critical-joint states. Therefore, if we consider each segment independently, the optimal path over this segment must always be associated with the edges on the CJSG, either in the form of decoupled paths on the base graph or by performing a supporting behavior. Thus, for any optimal path planned from the JSG, there is a path on CJSG with the same cost. This together with the fact that ${u}^{\dagger}$ is the optimal path on CJSG leads to the conclusion $Q({u}^{\dagger}) \le Q({u}^{\star})$. We complete the proof.
%since each segment has the minimum cost forming the JSG optimal path $u^{\star}$, where $u^{\star}$ can be mapped on CJSG as a path from initial state to goal state (i.e., $Q({u}^{\dagger}) \le Q({u}^{\star})$).
%If for any edge $\edge{pq,r\ell}$ in $u^{\star}$, no supporting behavior is performed on them. two agents are decoupled in the team planning for JSG. From lemma 1, the minimum cost in JSG is $Q(u^{\star})=R_{11,gg}=\psi_{1g}+\psi_{1g}$ which can be illustrated from CJSG planning since the initial state and goal state are connected in the CJSG. Once the supporting behaviour exists in $u^{\star}$, the supporting behaviour must happen at one critical joint state and end at the other critical joint state corresponding to the previous one, where this joint cost is same as the cost for these two states in the CJSG. This is equivalent to the path on which we set several critical joint states as the way-points. The rest parts of the path are just decoupled move planning. Those parts can be totally reflect on the CJSG planning since any two critical states are connect by the CJSG edge. 
\end{proof}

\subsection{Comparison of Computational Complexity}\label{Sec_complex}
 We quantify the computational complexity of the search algorithm applied to the joint state graph (JSG) and the critical joint state graph (CJSG), to demonstrate the advantage of CJSG over the JSG.

%Given a graph with $m$ nodes, the computational complexity of Dijkstra's search algorithm [x] is $\mathcal{O}(m^2)$. 
For JSG $\mathbb{J}=(\mathcal{S}, \mathcal{L})$ the graph construction complexity is given by the addition of complexities of nodes and edges. The complexity for the nodes is $\mathcal{O}(|\mathcal{S}|)=\mathcal{O}(|\mathcal{V}|^2)$. Similarly, the complexity of edges is $\mathcal{O}(|\mathcal{L}|)$ which equals $\mathcal{O}(|\mathcal{V}|^4)$ in worst case senario when edges are fully connected. Thus, the graph construction complexity of JSG equals
\begin{align}\label{eq_jsg_const}
    \mathcal{O}_{\text{Jconst}} = \mathcal{O}(|\mathcal{V}|^2)+ \mathcal{O}(|\mathcal{V}|^4).
\end{align}
Since the total number of nodes in JSG is $\mathcal{O}(|\mathcal{V}|^2)$, the search complexity when edges are fully connected follows
\begin{align}\label{eq_jsg_plan}
    \mathcal{O}_{\text{Jplan}} = \mathcal{O}(|\mathcal{V}|^4).
\end{align} 
Combining equations \eqref{eq_jsg_const} and \eqref{eq_jsg_plan}, complexity of JSG becomes
\begin{align}
    \mathcal{O}_{\text{JSG}} = \mathcal{O}(|\mathcal{V}|^4).
\end{align}

Similarly, the construction complexity of CJSG $\mathbb{T}=(\mathcal{M}, \mathcal{H})$ can be expressed as the addition of construction complexities for nodes and edges. For nodes, the complexity simply equals $\mathcal{O}(|\mathcal{M}|)$. For edges, the complexity equals $\mathcal{O}(|\mathcal{M}|^2)+\mathcal{O}(|\mathcal{V}|^2\log(|\mathcal{V}|))$, where the first term is the number of edges in $\mathbb{T}$, which is fully connected. The second term comes from Lemma \ref{Lm_base_graph}, which, in the worst case, needs to compute the shortest path between any pair of nodes in $\mathbb{G}$. The complexity of $\mathcal{O}(|\mathcal{V}|^2\log(|\mathcal{V}|))$ assumes the use of Johnson's algorithm \cite{CTH-LCE-RRL-SC:22}. Thus, the construction complexity of CJSG equals 
\begin{align}\label{eq_complx_cons}
    \mathcal{O}_{\text{const}}=\mathcal{O}(|\mathcal{M}|)+\mathcal{O}(|\mathcal{M}|^2)+\mathcal{O}(|\mathcal{V}|^2\log(|\mathcal{V}|)).
\end{align}
The search complexity of CJSG is determined by the number of nodes in $\mathbb{T}$, which follows
\begin{align}\label{eq_complx_search}
    \mathcal{O}_{\text{plan}}=\mathcal{O}(|\mathcal{M}|^2)+\mathcal{O}(|\mathcal{M}|).
\end{align}
where for the first term, we assume the use of Dijkstra's Algorithm \cite{CTH-LCE-RRL-SC:22} on $\mathbb{T}$, to obtain $\widehat{u}$. The second term is associated with reconstructing $u^{\dagger}$ from $\widehat{u}$. Although Algorithm \ref{Algorithm_search_CJSG} embeds search functions in step \ref{step_recons}, all the planning must have been computed by Lemma \ref{Lm_base_graph} when executing Algorithm \ref{Alg_defCJSG}. No replanning is needed. By combining \eqref{eq_complx_cons} and \eqref{eq_complx_search}, one has 
\begin{align}\label{eq_complx_cons2}
    \mathcal{O}_{\text{CJSG}}=\mathcal{O}(|\mathcal{M}|^2)+\mathcal{O}(|\mathcal{V}|^2\log(|\mathcal{V}|)).
\end{align}

\begin{remark}\longthmtitle{Comparison of complexity}\label{Rm_complexity}
To compare the complexities of $\mathcal{O}_{\text{CJSG}}$ and $\mathcal{O}_{\text{JSG}}$, we only need to compare $\mathcal{O}(|\mathcal{M}|^2)$ and $\mathcal{O}(|\mathcal{V}|^4)$. Note that in most scenarios, we assume the number of support edges in $\mathbb{G}$ is small. As a consequence, the number of critical-joint states is far less than that of common joint states, i.e. $|\mathcal{M}|\ll |\mathcal{V}|^2$. Then the proposed Algorithm \ref{Algorithm_search_CJSG} based on CJSG is significantly more efficient than the traditional JSG method. The worst boundary scenario happens when support edges widely exist in $\mathbb{G}$, in this case, one has $|\mathcal{M}|\to |\mathcal{V}|^2$, but $|\mathcal{M}|$ is still upper bounded by $|\mathcal{V}|^2$ due to the fact that critical joint states are subsets of joint states. Thus, $\mathcal{O}_{\text{CJSG}}$ is always no worse than $\mathcal{O}_{\text{JSG}}$.

%of $\mathcal{O}(|\mathcal{M}|)$  In the worst scenario, 
\end{remark}
%it with $\mathcal{O}_{\text{JSG}}=\mathcal{O}|(\mathcal{V}|^4)$, note that 
%The boundary conditions of $\mathcal{O}(|\mathcal{M}|)$ 

%can be explained by two extreme scenarios, one is all nodes are critical-joint states, and the other is no critical joint states which also means no support behavior in environment graph, i.e., CJSG = JSG. For the first scenario, $\mathcal{O}(|\mathcal{H}|)$ only depends on $|\mathcal{E}_R|$ and $\bar{N}$ since the edge computing does not need any information from the search on the ordinary graph, i.e., $\mathcal{O}(|\mathcal{H}|)=\mathcal{O}((|\mathcal{E}_R|\bar{N})^2)$. For the second scenario, the construction time complexity is equal to JSG's, i.e.,  $\mathcal{O}(|\mathcal{H}|)=\mathcal{O}|(\mathcal{V}|^4)$.  

%Compared the computational complexities between JSG and CJSG, the key idea of CJSG is to simplify the computation when the two agents' movements can be decoupled. If the number of critical joint states infinitely approaches the number of joint states, i.e., $||\mathcal{E}_R|\bar{N}|\to |\mathcal{V}|$, the gain of CJSG is limited. Otherwise, when $||\mathcal{E}_R|\bar{N}|\ll |\mathcal{V}|$, the effect of CJSG is more significant. 
%If $|\mathcal{E}_R|\bar{N}|\to |\mathcal{V}|$, the gain of CJSG is limited. This can be verified by XXX.
%When $|\mathcal{E}_R|\bar{N}|\ll |\mathcal{V}|$, the difference is significant. 
%Although we characterize the complexity as $\mathcal{O}((|\mathcal{E}_R|\bar{N}|)^2|\mathcal{V}|^2)$.  In practice, the computation can be further simplified by XXX. 
\section{NUMERICAL RESULTS}\label{sec:numerical-results}
% \daigo{We can show results from the same graph, but with different edge costs.  This was the original motivation of the research: ``how do we get different coordinated behavior based on the level or risk (or the cost of loosing an agent)?''}\\
 In this section, we evaluate JSG and CJSG on the basis of graph construction and path planning under different conditions.\footnote{\url{https://github.com/RobotiXX/team-coordination}} Our experimental design allows us to gain insights through comparative analysis of JSG and CJSG in terms of scalability and performance. The experiments are carried out on a MacBook Pro with 2.8 GHz 8 core CPU and 8GB of RAM.

%Agent B can also reduce the cost to $c_{41},_{44} = 3$ with support from agent A. However, in cases where supporting does not reduce the cost, such as in where edge $e_{1,4}$ has cost $c_{1,4}=3$, both agents traverse the edge together without supporting each other.
%\manshi{Sara: This needs to be updated based on illustrative example}

For both graph construction and path planning analyses, a random graph generator is used to generate environment graphs with varying number of nodes and edges. 
% where nodes $v \in \nodeset$ and edges $\edgeset \subset \nodeset \times \nodeset$. 
We control the ratio of risk edges to the total edges to be 1/5, 1/3, and 1/2. For different number of nodes and risk edges ratio in an environment graph, we calculate graph construction time and shortest path planning time for JSG and CJSG (Table~\ref{table:result_table}). We repeat every experiment trial five times for statistical significance.

%\begin{figure}[H]
%    \includegraphics[width=0.5\columnwidth]{figures/cost cases.pdf}
%    \caption{Illustrative example with two cases for different risk edge costs.}
%    \label{fig_illustrative_solution1}
%\end{figure}
%\begin{figure}[H]
%    \includegraphics[width=\columnwidth]{figures/Lower risk cost.pdf}
%    \caption{Illustrative example without support}
%    \label{fig_illustrative_solution2}
%\end{figure}

% \begin{figure}[H]
%     \includegraphics[width=\columnwidth]{figures/TT_JSGvsCJSG.png}
%     \caption{Comparison of total solution time taken by JSG and CJSG with respect to increasing number of nodes.}
%     \label{fig_wrt_node}
% \end{figure}

% \begin{figure}[H]
%     \includegraphics[width=\columnwidth]{figures/TT_wrt_riskyedgeratio.png}
%     \caption{Comparison of total solution time taken by JSG and CJSG with respect to increasing number of nodes.}
%     \label{fig_wrt_riskedgratio}
% \end{figure}

% \begin{figure}[H]
%     \includegraphics[width=\columnwidth]{figures/FinalJSGvsCJSGPlot.pdf}
%     \caption{Comparison of time taken by JSG and CJSG to generate total solution with respect to increasing number of nodes and risk edges ratio.}
%     \label{fig_wrt_riskedgratio}
% \end{figure}

\begin{figure}[t]
    \centering
    \includegraphics[width=0.5\textwidth]{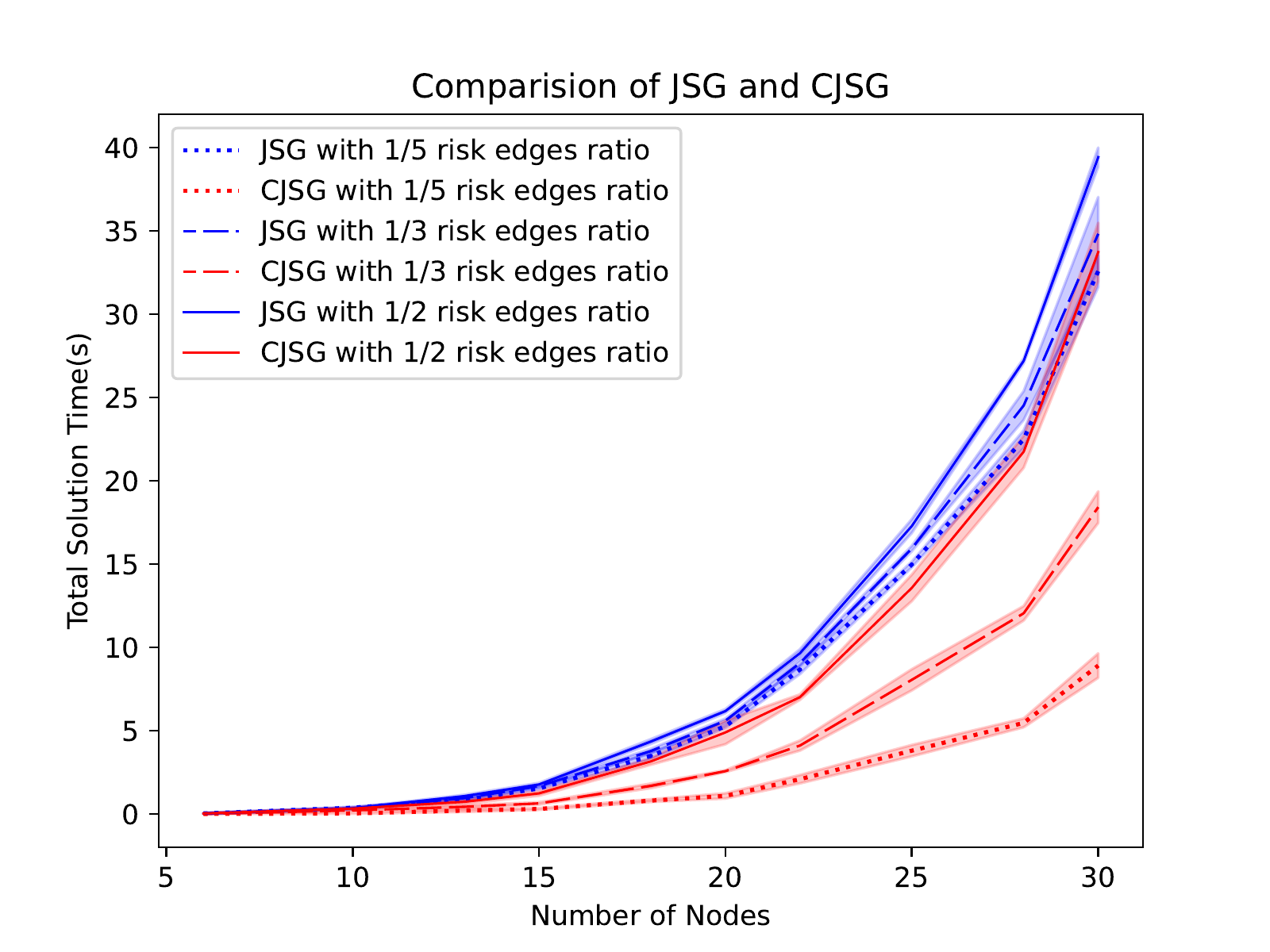}
    \caption{Comparison of time taken by JSG and CJSG to generate total solution with respect to increasing number of nodes and risk edges ratio.}
    \label{fig_wrt_riskedgratio}
\end{figure}

\begin{table*}[ht]
 \caption{Comparison of JSG and CJSG}\label{table:result_table}
\centering
\begin{tabularx}{\textwidth}{@{} l *{10}{C} c @{}}
\toprule
\multicolumn{1}{c}{} & \multicolumn{1}{c}{} & \multicolumn{2}{c}{\textbf{JSG}} & \multicolumn{2}{c}{\textbf{CJSG}} \\
\cmidrule(rl){3-4} \cmidrule(rl){5-6}
\textbf{Nodes} &  \textbf{Risk Edges Ratio} &  {Graph Construction(s)} & {Shortest Path(s)} & {Graph Construction(s)} & {Shortest Path(s)} \\

\midrule
%10 & 1/5 &  $0.4576 \pm 0.0997$ & $0.3978\pm0.1229$ & $0.0618\pm0.0314$ & $0.1428\pm$
10   & 1/5  & 0.2119$\pm$0.0410 & 0.1440$\pm$0.0176 & 0.0146$\pm$0.0021 & 0.0198$\pm$0.0152 \\
     &  1/3  &  0.1760$\pm$0.0519 & 0.1510$\pm$0.0040 & 0.0895$\pm$0.0151 & 0.1258$\pm$0.0207 \\
      & 1/2 & 0.1906$\pm$0.0273 & 0.1567$\pm$0.0091 & 0.1810$\pm$0.0143 & 0.1469$\pm$0.0478 \\

\midrule 
20 & 1/5 &  3.1662$\pm$0.0405 & 2.098$\pm$0.0445 & 0.6525$\pm$0.0931 & 0.4348$\pm$0.0651 \\
    & 1/3 & 3.3989$\pm$0.0603 & 2.1988$\pm$0.0660 & 1.7742$\pm$0.0497 & 0.7988$\pm$0.0094 \\
    & 1/2 & 3.8566$\pm$0.0906 & 2.3192$\pm$0.0265 & 3.6439$\pm$0.5942 & 1.2523$\pm$0.1181 \\
    
\midrule 
30 & 1/5 & 20.9363$\pm$0.8312 & 11.6431$\pm$0.12776 & 6.1126$\pm$0.5537 & 2.7973$\pm$0.1778 \\
    & 1/3 & 22.4891$\pm$1.7074 & 12.3330$\pm$0.4995 & 13.8171$\pm$0.7259 & 4.5996$\pm$0.2204 \\
    & 1/2 & 25.9774$\pm$0.4323 & 13.4440$\pm$0.14222 & 26.8156$\pm$1.49423 & 6.9094$\pm$0.2677 \\  
\bottomrule
\end{tabularx}
\end{table*} 

\subsection{Graph Construction Analysis} 
From Table~\ref{table:result_table}, we analyze the graph construction time for JSG and CJSG under different conditions. Given a fixed risk edges ratio, e.g., 1/3 of the total edges, the improvement in graph construction time by CJSG compared to JSG maintains as the number of nodes increases from 10 to 30. Similarly, if we fix the number of nodes, e.g., 10, and increase the risk edges ratio gradually from 1/5, then 1/3, and finally 1/2, CJSG still takes less time compared to JSG. We can also see such a pattern for node 20 and node 30. 
%(with an exception at ratio 1/2). 
These results provide empirical evidence that CJSG is more efficient in constructing graphs. Note that when the risk edge ratio reaches 1/2, nearly all joint states are critical joint states, i.e., $|\mathcal{M}|\to|\mathcal{V}|^2$, and the graph construction times for the two approaches are close to each other. This observation is in line with Remark \ref{Rm_complexity}.

 % We conduct two sets of experiments to analyze graph construction time. In the first set of experiments, where the number of risk edges is fixed to be 1/3 of the total edges, CJSG performs slightly better than JSG as the number of nodes increases, with JSG taking marginally more time to construct the graph (as shown in Table~\ref{table:result_table}). In the second set of experiments, we compare the construction time of of different ratio of risk edges to node but with a fixed size of nodes and edges. Table~\ref{table:result_table} illustrates that as the ratio of risk edges to the number of nodes increases, CJSG takes slightly less construction time than JSG.  

\subsection{Path Planning Analysis} 

From Table~\ref{table:result_table}, we also assess the path planning time for JSG and CJSG with varying nodes and risk edges ratio. Given a certain risk edges ratio, e.g., 1/3, we can see that CJSG takes less time than JSG. This is true even if we increase the nodes from 10 to 30. Similarly, if we fix the node size, e.g., 20, and gradually increase the risk edges ratio as 1/5, 1/3, and 1/2 of total edges, CJSG is still more efficient than JSG. We can see the same pattern for nodes 10, 20 and 30. These results indicate that CJSG is more efficient than JSG in terms of shortest path planning when the ratio of risk edges to nodes increases. 

 % Our path planning analysis involved comparing the shortest path planning time between JSG and CJSG using two experiments. For the first experiment, we assumed the size of risk edges to be one third of the total edges for any given node. We then evaluated the performance of JSG and CJSG with respect to an increasing number of nodes. Table~\ref{table:result_table} illustrates that as the number of nodes increases, the path planning time for JSG increases more rapidly in comparison to CJSG. In the second experiment, we compared the shortest path planning time of JSG and CJSG with varying ratios of risk edges per node. The experiment involved keeping the number of nodes and edges constant while gradually increasing the number of risk edges ratio. Table~\ref{table:result_table} shows that as the ratio increases, JSG's path planning time becomes significantly larger than that of CJSG. These results indicate that CJSG is more efficient than JSG in  terms of shortest path planning when the ratio of risk edges to nodes increases. 

Based on the experimental results shown in Table~\ref{table:result_table}, we compute the total time taken by both JSG and CJSG to find the final solution. The total time involves time taken for graph construction and shortest path planning. In Fig. \ref{fig_wrt_riskedgratio}, we show that as the number of nodes increases, the time to generate total solution for JSG increases more significantly than that of CJSG. Fig. \ref{fig_wrt_riskedgratio} also illustrates that as the risk edges ratio increases, the time to generate solution for JSG increases more significantly compared to CJSG. Thus, CJSG is more efficient than JSG for overall solution generation.

\section{CONCLUSIONS}\label{sec:conclusion}

 We presented a team coordination problem in a graph environment, where high levels of coordination in the form of “support” allows agents to reduce the cost of traversal on some edges. As an alternative to solving this with a version of MDP, we presented a method of planning in the joint state space – the Joint State Graph (JSG). We showed that a multi-agent path planning problem can be reduced to a single-agent planning in JSG, since the actions taken by the agents are built in to the edges of the JSG. We addressed the issue of scalability in the graph size by presenting a hierarchical decomposition method to perform path planning in two levels. We provided complexity and statistical analysis which show that the construction time for both CJSG and JSG do not differ by much, but the CJSG is significantly more efficient with regards to shortest path planning. Our numerical results verify this.

For future work, there are many aspects of the problem we proposed that we are intrigued to build upon. For instance, we would like to integrate more sophisticated notions of risk by using concepts from game theory and incorporating stochasticity in the formulation, such as stochastic costs. We are also interested in addressing the issue of scalability in terms of number of agents.

\bibliographystyle{ieeetr}
\bibliography{bib}
\end{document}